\documentclass[12pt, final]{article}

\usepackage{amsmath, amsthm, amssymb, amsfonts}
\DeclareMathOperator{\topop}{top}
\DeclareMathOperator{\TR}{TR}
\DeclareMathOperator{\IR}{IR}
\DeclareMathOperator{\TTR}{TTR}
\DeclareMathOperator{\T2TR}{T2TR}
\DeclareMathOperator{\PR}{PR}

\usepackage[letterpaper,margin=1.25in]{geometry} % matches econsocart draft text block (6in x 8.5in)
\usepackage{setspace}
\usepackage[dvipsnames]{xcolor}
\usepackage{tikz}
\usepackage{subfiles}
\usepackage{ifthen}

\theoremstyle{plain}
\newtheorem{theorem}{Theorem}
\newtheorem{lemma}{Lemma}
\newtheorem{proposition}{Proposition}
\newtheorem{corollary}{Corollary}
\theoremstyle{definition}
\newtheorem{example}{Example}
\newtheorem{definition}{Definition}

\usepackage{natbib}
\usepackage[
    colorlinks=true,
    linkcolor=blue,
    citecolor=blue,
    urlcolor=blue,
    filecolor=blue,
    hyperfootnotes=false,
    hyperindex,
    breaklinks
]{hyperref}

\usepackage{newtxtext} %\usepackage{times}
\usepackage{newtxmath}
\usepackage{microtype}
\microtypesetup{expansion=false} % avoid squeezing extra characters per line vs the class file

\usepackage[normalem]{ulem}

\title{Minimally rational reallocation of objects%
  \thanks{We used ChatGPT to assist with editing, to check proofs, and to construct the rule in \autoref{ex:fails-tr-pr}.}}

\author{
\"Ozg\"un Ekici\thanks{Department of Economics, \"Ozye\u{g}in University, Istanbul, T\"urkiye. Email: \texttt{ozgun.ekici@ozyegin.edu.tr}}
\and
M. Bumin Yenmez\thanks{Department of Economics, Washington University in St. Louis, USA; Durham University, UK; and \"Ozye\u{g}in University, T\"urkiye. Email: \texttt{bumin@wustl.edu}}
}
\date{September 23, 2026}

\usepackage{lineno}

\begin{document}
%\linenumbers
\maketitle

\begin{abstract}
Matching theory has largely evolved around two canonical rules: deferred acceptance (DA) in the marriage problem and top trading cycles (TTC) in the object reallocation problem. Although the two rules operate through different procedures, we show that they rest on a common axiomatic foundation. In the marriage problem, stability decomposes into individual rationality and pair rationality, and a classic result characterizes DA by these two axioms together with strategy-proofness for the proposing side. We show that individual rationality, pair rationality, and strategy-proofness likewise characterize TTC. More strongly, three weak rationality axioms concerning individuals and pairs, together with strategy-proofness, uniquely determine TTC. Even weak rationality requirements therefore rule out strategy-proof implementation when there is a cap on the size of exchange cycles.
\end{abstract}

\begingroup
\endgroup

%%%%%%%%%%%%%%%%%%%%%%%%%%%%%%%%%%%%%%%%%%%%%%%%%%%%%%%%%%%%%%%%%%%%%%%%%
%%%% Main text entry area:
%%%%%%%%%%%%%%%%%%%%%%%%%%%%%%%%%%%%%%%%%%%%%%%%%%%%%%%%%%%%%%%%%%%%%%%%%

\section{Introduction}\label{sec:intro}

Deferred acceptance (DA) and top trading cycles (TTC) are two fundamental mechanisms in the theory of matching markets. Their origins lie in the marriage problem introduced by \citet{Gale1962AMM} and the object reallocation problem introduced by \citet{Shapley1974JME}. In the marriage problem, agents on both sides of the market have preferences over potential partners. In the object reallocation problem, each agent initially owns one object and has preferences over objects. Although DA and TTC operate through different procedures and are usually motivated by different solution concepts, we show that they share a common axiomatic foundation: each is the unique rule in its model satisfying individual rationality, pair rationality, and strategy-proofness, which in the marriage problem is required only of the proposing side. In the object reallocation problem, we obtain this characterization from substantially weaker rationality requirements concerning the first and second choices of individuals and pairs of agents.

In the marriage problem, \citet{Gale1962AMM} proposed DA, an iterative procedure in which agents on one side make proposals and agents on the other side tentatively retain their most preferred proposals. They showed that DA selects the man-optimal stable matching when men propose. Their original model ruled out remaining single, but the result extends to a setting with this option. Stability then combines individual rationality (IR), which requires that no agent be assigned a partner ranked below remaining single, and pair rationality (PR), which requires that no man and woman both prefer one another to their assigned partners. \citet{Roth1982MOR} showed that DA is strategy-proof (SP) for the proposing agents. \citet{Alcalde1994ET} established that SP for the proposing side, IR, and PR uniquely characterize DA.

We establish the corresponding characterization for TTC in the object reallocation problem. \citet{Shapley1974JME} proposed TTC, which they attribute to David Gale. The rule reallocates objects through trading cycles: at each stage, every remaining agent points to her favorite remaining object, every remaining object points to its owner, and all cycles are executed. Here, IR requires that each agent receive an object she finds at least as desirable as her endowment, while PR requires that no two agents can exchange their endowments so that one is strictly better off and the other no worse off. We show that TTC is the unique rule satisfying SP, IR, and PR (\autoref{cor:irpr}). Thus, the same combination of incentive and rationality requirements that determines DA in the two-sided model determines TTC in the one-sided model.

The force of this result comes from the distinction between bilateral rationality and core membership. \citet{Shapley1974JME} showed that the TTC allocation belongs to the core: no coalition can reallocate its endowments to make all its members at least as well off and at least one strictly better off. This core allocation is unique \citep{Roth1977JME}. Core membership therefore imposes restrictions on coalitions of every size, whereas IR and PR concern only individuals and pairs. Our result shows that, together with SP, these individual and bilateral requirements already determine the entire core allocation.

Previous characterizations of TTC combine SP and IR with various third axioms: Pareto efficiency in \citet{Ma1994IJGT}, endowments-swapping-proofness in \citet{Fujinaka2018GEB}, and pair efficiency in \citet{Ekici2024TE}.\footnote{Pair efficiency concerns exchanges of assignments, whereas PR concerns exchanges of endowments. The two are logically incomparable, even among individually rational allocations (\autoref{ex:fails-t2tr}).} A characterization based on pair rationality, however, had not been obtained. Our result fills this gap and establishes the parallel with DA. In fact, we prove a stronger result: TTC is the unique rule satisfying SP and three weak rationality axioms—top rationality (TR), top-top rationality (TTR), and top2-top rationality (T2TR) (\autoref{thm:characterization}). The characterization by SP, IR, and PR follows as a corollary (\autoref{cor:irpr}).

These weak rationality axioms concern individual assignments and bilateral exchanges involving agents' first and second choices. TR requires that an agent receive her endowment whenever it is her top choice. TTR requires two agents to receive each other's endowments whenever each ranks the other's endowment first. T2TR applies when one agent ranks the other's endowment first and the other ranks the first agent's endowment second: if the latter does not receive her first choice, the two agents must receive each other's endowments. These conditions protect an agent's endowment when it is her top choice and require particular bilateral trades when the agents rank one another's endowments sufficiently highly.

Comparing these three rationality axioms with IR and PR clarifies the role of strategy-proofness. TR is implied by IR and, at a fixed profile, is strictly weaker (\autoref{prop:tr-weakerthan-ir}). Among SP rules, however, the two are equivalent: SP and TR together imply IR (\autoref{prop:sp-tr-implies-ir}). TTR and T2TR are implied by PR and, even in conjunction, are strictly weaker at a fixed profile (\autoref{prop:ttr-t2tr-weakerthan-pr}). The difference survives among SP rules: there are SP rules satisfying TTR and T2TR that violate PR (\autoref{prop:sp-ttr-t2tr-notimplies-pr}). Nevertheless, when all three rationality axioms are combined with SP, they determine TTC and hence recover full core membership.

DA and TTC differ in their procedures and in the strategic structure of their markets, yet each is uniquely determined by IR, PR, and strategy-proofness for the relevant side: the proposing agents under DA and all agents under TTC.\footnote{In the marriage problem,  no rule satisfies IR and PR while being strategy-proof for both sides \citep{Roth1982MOR}. However, SP for one side amounts to full strategy-proofness if the other side consists of resources with exogenous priority rankings instead of agents with preferences, as in the school-choice setting.} In the object reallocation model, our stronger characterization shows that even weaker rationality conditions suffice: protecting an agent's endowment when it is her top choice and requiring two specific forms of bilateral exchange.

The characterization also has implications for restrictions on trading cycles. In applications such as kidney exchange, logistical constraints make large exchanges undesirable. In our object reallocation model, TTC may require arbitrarily large cycles as the market grows. Our main theorem therefore implies that no rule satisfies SP together with TR, TTR, and T2TR once exchange cycles are capped below the size of the market (\autoref{cor:restricted-cycles}). For two-way exchange the obstruction is sharper still: T2TR alone cannot be satisfied at a particular three-agent profile, even without imposing TR and SP (see the discussion following \autoref{cor:restricted-cycles}).

The rest of the paper is organized as follows. \autoref{sec:model} presents the model, and \autoref{sec:ttc} presents TTC. \autoref{sec:rationality-axioms} introduces and compares the rationality axioms. \autoref{sec:main-result} states the main result, and \autoref{sec:proof} provides its proof. \autoref{sec:sum-rellit} summarizes the paper and discusses the related literature. The independence of the axioms in our main result is established in the \hyperref[sec:indep-of-axioms]{Appendix}.

\section{Model}\label{sec:model}

Let $N=\{1,2,\ldots,n\}$ be a finite set of agents. Let $O=\{o_1,o_2,\ldots,o_n\}$ be a set of indivisible objects. Each agent $i$ initially owns object $o_i$, called her \textbf{endowment}.

For any $X\subseteq N$, let $\omega(X)=\{o_i:i\in X\}$ denote the set of objects initially owned by members of $X$. A nonempty subset of agents is called a \textbf{coalition}.

Each agent cares only about the object she receives. Let
$\mathcal{P}$ be the set of all strict linear orders on $O$.
Agent $i$'s preference relation is denoted by $P_i\in\mathcal{P}$,
where $a\,P_i\,b$ means that she prefers object $a$ to object $b$.
Let $R_i$ denote the associated weak preference relation:
$a\,R_i\,b$ if either $a\,P_i\,b$ or $a=b$.
When convenient, we present preferences by listing objects
from most preferred to least preferred. For example,
$P_i:o_2\succ o_1\succ o_3$ means
$o_2\,P_i\,o_1\,P_i\,o_3$.

For $q\in\{1,\ldots,n\}$, let $\topop_q(P_i)$ denote
the object ranked $q^{\text{th}}$ under $P_i$.

A \textbf{preference profile} is a list
$P=(P_i)_{i\in N}\in\mathcal{P}^n$.
For $X\subseteq N$, the mixed profile
$(P_X,\tilde{P}_{N\setminus X})$ assigns preference relation
$P_i$ to each $i\in X$ and $\tilde{P}_i$ to each
$i\in N\setminus X$. When $X=\{i\}$, we write
$(P_i,\tilde{P}_{-i})$. Other mixed-profile notation
is understood accordingly.

An \textbf{allocation} is a bijection $\mu:N\to O$.
Let $\mathcal{M}$ denote the set of allocations, and let
$\mu(i)$ denote agent $i$'s assignment under $\mu$.
When convenient, we represent an allocation as an ordered
list in square brackets. For example, with three agents,
$\mu=[o_2,o_1,o_3]$ means that $\mu(1)=o_2$,
$\mu(2)=o_1$, and $\mu(3)=o_3$. A \textbf{rule} is a mapping
$\phi:\mathcal{P}^n\to\mathcal{M}$.
Let $\phi_i(P)$ denote agent $i$'s assignment under
the allocation $\phi(P)$.

An allocation $\nu$ \textbf{Pareto dominates} an allocation
$\mu$ at profile $P$ if $\nu(i)\,R_i\,\mu(i)$ for every
$i\in N$, with strict preference for at least one agent.
An allocation is \textbf{Pareto-efficient} at $P$ if no
allocation Pareto dominates it at $P$. A rule is
Pareto-efficient if it selects a Pareto-efficient allocation
at every profile.

An allocation $\mu$ is \textbf{pair efficient} at $P$ if
there are no two agents $i,j\in N$ such that
$\mu(j)\,P_i\,\mu(i)$ and $\mu(i)\,P_j\,\mu(j)$. Thus, under
a pair efficient allocation, no two agents can exchange
their assignments to the benefit of both.

An allocation $\mu$ is a \textbf{core allocation} at profile $P$ if there is no coalition $X$ and bijection $\nu:X\to\omega(X)$ such that $\nu(i)\,R_i\,\mu(i)$ for every $i\in X$, with strict preference for at least one member of $X$.

Our incentive compatibility notion is strategy-proofness, which rules out an agent benefiting from misreporting her preferences. A rule $\phi$ is \textbf{strategy-proof (SP)} if, for every profile $P\in\mathcal{P}^n$, agent $i\in N$, and report $\tilde{P}_i\in\mathcal{P}$, $\phi_i(P)\,R_i\,\phi_i(\tilde{P}_i,P_{-i})$. Thus, truthful reporting is a weakly dominant strategy.

In the remainder of the paper, we often abbreviate the names of axioms to simplify the exposition. For example, the reader should read ``$\phi$ satisfies SP'' as ``$\phi$ satisfies strategy-proofness.'' The same convention applies to the rationality axioms introduced in \autoref{sec:rationality-axioms}.

\section{Top trading cycles}\label{sec:ttc}

A \textbf{cycle} consists of distinct agents $i_1,\ldots,i_s$ and their endowments, with each agent pointing to the next agent's endowment, the last agent pointing to the first agent's endowment, and each object pointing to its owner. We represent this cycle by $(i_1,\ldots,i_s)$; the representation is unique up to rotation, so we may start it at any of its agents. Graphically, it looks as follows:

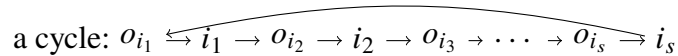
\begin{center}
\begin{tikzpicture}[scale=1]
\node at (-1,0) {a cycle:};
\node (o1) at (0,0) {$o_{i_1}$};
\node (i1) at (1,0) {$i_1$};
\node (o2) at (2,0) {$o_{i_2}$};
\node (i2) at (3,0) {$i_2$};
\node (o3) at (4,0) {$o_{i_3}$};
\node (dots) at (5,0) {$\dots$};
\node (os) at (6,0) {$o_{i_s}$};
\node (is) at (7,0) {$i_s$};

\draw (o1) edge [->] (i1);
\draw (i1) edge [->] (o2);
\draw (o2) edge [->] (i2);
\draw (i2) edge [->] (o3);
\draw (o3) edge [->] (dots);
\draw (dots) edge [->] (os);
\draw (os) edge [->] (is);
\draw (is) edge [->, bend right=12] (o1);
\end{tikzpicture}
\end{center}

The \textbf{size of a cycle} is the number of agents it contains. A cycle may have size one: in the cycle $(i)$, agent $i$ points to her own endowment.

A cycle $(i_1,\ldots,i_s)$ is \textbf{executed} when each agent receives the object to which she points. Thus, an allocation $\mu$ executes this cycle if
\[
\mu(i_r)=o_{i_{r+1}}
\quad\text{for }r=1,\ldots,s-1,
\qquad
\mu(i_s)=o_{i_1}.
\]

A cycle $(i_1,\ldots,i_s)$ is an \textbf{all-tops cycle} at profile $P$ if each agent in the cycle ranks the next agent's endowment first, where the agent following $i_s$ is $i_1$.

\begin{description}

\item[\textbf{\textsc{Top Trading Cycles}}]

\item[]
Given a preference profile $P$, construct an allocation
as follows.

\item[\textbf{Step 1:}]
Construct a directed graph whose nodes are all agents and
objects. Each agent points to her favorite object, and each
object points to its owner. Since the graph is finite and
each node has exactly one outgoing edge, it contains at
least one cycle, and its cycles are disjoint. Execute all
cycles and remove their agents and objects. If any agents
remain, proceed to Step 2. Otherwise, terminate.

\item[\textbf{Step $t\geq2$:}]
Construct a directed graph whose nodes are the remaining
agents and objects. Each remaining agent points to her
favorite remaining object, and each remaining object points
to its owner. As in Step 1, the graph contains at least
one cycle, and its cycles are disjoint. Execute all cycles
and remove their agents
and objects. If any agents remain, proceed to Step $t+1$.
Otherwise, terminate.

\end{description}

We use $\varphi^{\mathrm{ttc}}$ to denote the TTC rule. It is
Pareto-efficient and strategy-proof \citep{Roth1982EL} and, at each preference profile, selects the unique core allocation \citep{Shapley1974JME,Roth1977JME}.

\section{Rationality axioms}\label{sec:rationality-axioms}

Suppose an allocation $\mu$ is proposed to the agents. For a coalition $X$, it is not \emph{rational} to accept the proposed allocation if its members can reallocate their endowments so that every member is weakly better off and at least one is strictly better off. In this sense, the core imposes a stringent rationality condition: the proposed allocation must be rational for every coalition. In this section, we consider weaker notions of rationality.

The rationality axioms below are defined for an allocation at a given preference profile. We say that a rule satisfies one of these axioms if, at every profile, the allocation selected by the rule satisfies that axiom. For example, a rule $\phi$ is individually rational if, for every profile $P$ and every agent $i$, $\phi_i(P)\,R_i\,o_i$.

For individuals, we consider two rationality axioms. Individual rationality requires an agent's assignment to be at least as good as her endowment. Top rationality requires that an agent be assigned her endowment whenever it is her top choice.

\begin{definition}[IR, TR]\label{def:ir-tr}
An allocation $\mu$ is:
\begin{enumerate}
    \item[(a)] \textbf{individually rational (IR)} at $P$ if, for each $i\in N$, $\mu(i)\,R_i\,o_i$.

    \item[(b)] \textbf{top rational (TR)} at $P$ if, for each $i\in N$ with $\topop_1(P_i)=o_i$, $\mu(i)=o_i$.
\end{enumerate}
\end{definition}

IR implies TR. For $n\geq3$, the converse fails: a TR allocation need not be IR. The following proposition illustrates how much less restrictive TR can be by comparing the numbers of allocations satisfying the two axioms.

Let $\IR(P)$ and $\TR(P)$ denote, respectively, the sets of IR and TR allocations at profile $P\in\mathcal{P}^n$.

\begin{proposition}\label{prop:tr-weakerthan-ir}
For $n\geq2$, there exists $P\in\mathcal{P}^n$ such that
\[
\frac{|\TR(P)|}{|\IR(P)|}=\frac{n!}{2}.
\]
\end{proposition}

\begin{proof}
Consider a profile $P$ with
\[
P_i:o_{i+1}\succ o_i\succ\cdots
\quad\text{for }i=1,\ldots,n-1,
\qquad
P_n:o_1\succ o_n\succ\cdots,
\]
where the remaining objects are ranked arbitrarily.

Since no agent ranks her endowment first, all $n!$
allocations satisfy TR. Under IR, each agent can receive
only her top choice or her endowment. If an agent receives
her endowment, the preceding agent in the cycle
$(1,\ldots,n)$ cannot receive her top choice and must
therefore receive her own endowment. Repeating this
argument around the cycle shows that either every agent
receives her endowment or every agent receives her top
choice. Thus, there are exactly two IR allocations:
$[o_1,o_2,\ldots,o_n]$ and $[o_2,o_3,\ldots,o_n,o_1]$.
Therefore, $|\TR(P)|/|\IR(P)|=n!/2$.
\end{proof}

Although TR can be much weaker than IR, the following proposition shows that TR implies IR when combined with SP.

\begin{proposition}\label{prop:sp-tr-implies-ir}
  A rule that satisfies SP and TR also satisfies IR.
\end{proposition}

\begin{proof}
Suppose that $\phi$ violates IR at some profile $P$,
so that $o_i\,P_i\,\phi_i(P)$ for some agent $i$.
Consider the report $\tilde{P}_i$ that ranks $o_i$ first.
By TR, $\phi_i(\tilde{P}_i,P_{-i})=o_i$, which agent $i$
prefers to $\phi_i(P)$. This is a profitable manipulation,
contradicting SP.
\end{proof}

One way to read \autoref{prop:sp-tr-implies-ir} is as a demonstration of the strength of SP: even when paired with a weak axiom, it entails a substantially stronger one.

For pairs of agents, we consider three rationality axioms. Pair rationality requires that no pair of agents can become better off in the Pareto sense by exchanging their endowments, compared with their proposed assignments. Top-top rationality requires two agents to exchange their endowments whenever each ranks the other's endowment first. Top2-top rationality requires that, when one agent ranks the other's endowment second and the other ranks the first agent's endowment first, either the first agent receives her top choice or the two agents exchange their endowments.

\begin{definition}[PR, TTR, T2TR]\label{def:pr-ttr-t2tr}
An allocation $\mu$ is:
\begin{enumerate}
    \item[(a)] \textbf{pair rational (PR)} at $P$
    if there are no distinct agents $i,j\in N$ such that $o_j\,P_i\,\mu(i)$ and $o_i\,R_j\,\mu(j)$.

    \item[(b)] \textbf{top-top rational (TTR)} at $P$ if, for every two distinct agents $i,j\in N$, whenever $\topop_1(P_i)=o_j$ and $\topop_1(P_j)=o_i$, we have $\mu(i)=o_j$ and $\mu(j)=o_i$.

    \item[(c)] \textbf{top2-top rational (T2TR)} at $P$
    if, for every ordered pair $(i,j)\in N^2$ with
    $i\neq j$, whenever $\topop_2(P_i)=o_j$ and
    $\topop_1(P_j)=o_i$, either
    $\mu(i)=\topop_1(P_i)$ or both
    $\mu(i)=o_j$ and $\mu(j)=o_i$.
\end{enumerate}
\end{definition}

Together, IR and PR are exactly the core requirement restricted to coalitions of size at most two. For a singleton $\{i\}$, the only reallocation assigns $o_i$ to agent $i$, and it blocks $\mu$ precisely when $o_i\,P_i\,\mu(i)$; ruling this out is IR. For a pair $\{i,j\}$, the reallocation that leaves both endowments in place blocks only if IR fails, while the exchange of endowments blocks precisely when PR fails.

While PR implies both TTR and T2TR,\footnote{Suppose $\mu$ is PR. For TTR, let $\topop_1(P_i)=o_j$ and $\topop_1(P_j)=o_i$. If $\mu(i)\neq o_j$, then $o_j\,P_i\,\mu(i)$, while $o_i\,R_j\,\mu(j)$ holds because $o_i$ is agent $j$'s top choice; so $i$ and $j$ block, a contradiction. Hence $\mu(i)=o_j$. If $\mu(j)\neq o_i$, then $o_i\,P_j\,\mu(j)$ and $o_j\,R_i\,\mu(i)$, the latter with equality; again $i$ and $j$ block. Hence $\mu(j)=o_i$. For T2TR, let $\topop_2(P_i)=o_j$ and $\topop_1(P_j)=o_i$, and suppose $\mu(i)\neq\topop_1(P_i)$. If also $\mu(i)\neq o_j$, then $\mu(i)$ is ranked below $o_j$ by agent $i$, so $o_j\,P_i\,\mu(i)$, while $o_i\,R_j\,\mu(j)$ holds as before; so $i$ and $j$ block. Hence $\mu(i)=o_j$, and the argument just given yields $\mu(j)=o_i$.} the following example shows that TTR and T2TR are logically independent.

\begin{example}[TTR vs T2TR]\label{ex:ttr-t2tr}
Suppose there are three agents. Consider the profiles
\[
P^*:\!
\left\{
\begin{aligned}
P_1^* &: o_2 \succ o_3 \succ \boxed{o_1},\\
P_2^* &: o_3 \succ o_1 \succ \boxed{o_2},\\
P_3^* &: o_1 \succ o_2 \succ \boxed{o_3},
\end{aligned}
\right.
\qquad
P^+:\!
\left\{
\begin{aligned}
P_1^+ &: \boxed{o_2} \succ o_1 \succ o_3,\\
P_2^+ &: o_1 \succ \boxed{o_3} \succ o_2,\\
P_3^+ &: \boxed{o_1} \succ o_3 \succ o_2.
\end{aligned}
\right.
\]

Let $\mu^*$ and $\mu^+$ be the allocations indicated
with boxes at $P^*$ and $P^+$, respectively.
At $P^*$, TTR holds vacuously because no two agents rank
each other's endowments first. However, T2TR fails:
agent 1 ranks $o_3$ second and agent 3 ranks $o_1$ first,
yet agent 1 receives neither her top choice nor $o_3$.

At $P^+$, T2TR holds vacuously because no ordered pair
satisfies its premise. However, TTR fails: agents 1 and 2
rank each other's endowments first, yet agent 2 receives
$o_3$.

To establish independence for rules, consider the following rules:
\[
\phi^*(P)=
\begin{cases}
\mu^*, & \text{if } P=P^*,\\
\varphi^{\mathrm{ttc}}(P), & \text{otherwise},
\end{cases}
\qquad
\phi^+(P)=
\begin{cases}
\mu^+, & \text{if } P=P^+,\\
\varphi^{\mathrm{ttc}}(P), & \text{otherwise}.
\end{cases}
\]
The rule $\phi^*$ satisfies TTR but violates T2TR, whereas $\phi^+$ satisfies T2TR but violates TTR.
\end{example}

Although PR implies both TTR and T2TR, their conjunction does not imply PR in general. Comparing the sizes of the corresponding sets of allocations at a fixed profile illustrates that TTR and T2TR, even jointly, are much weaker than PR.

Let $\TTR(P)$, $\T2TR(P)$, and $\PR(P)$ denote,
respectively, the sets of TTR, T2TR, and PR allocations
at profile $P\in\mathcal{P}^n$.

\begin{proposition}\label{prop:ttr-t2tr-weakerthan-pr}
For $n\geq2$, there exists $P\in\mathcal{P}^n$ such that
\[
\frac{|\TTR(P)\cap\T2TR(P)|}{|\PR(P)|}=(n-2)!.
\]
\end{proposition}

\begin{proof}
Group the objects into successive pairs
$(o_1,o_2),(o_3,o_4),\ldots$.
All agents rank every object in an earlier pair above
every object in a later pair. Within each pair,
odd-indexed agents prefer the even-indexed object,
while even-indexed agents prefer the odd-indexed object.
If $n$ is odd, the unpaired object $o_n$ is ranked last
by every agent. Thus, the rankings follow the pattern
\[
P:\!
\left\{
\begin{aligned}
P_i &: o_2\succ o_1\succ o_4\succ o_3\succ\cdots,
&& \text{if }i\text{ is odd},\\
P_i &: o_1\succ o_2\succ o_3\succ o_4\succ\cdots,
&& \text{if }i\text{ is even},
\end{aligned}
\right.
\]
where the displayed pattern includes only complete
pairs present in $O$, followed by $o_n$ if $n$ is odd.

Agents 1 and 2 are the only pair who rank each other's
endowments first. Moreover, no ordered pair satisfies
the premise of T2TR. Thus, an allocation satisfies both
TTR and T2TR if and only if agents 1 and 2 receive each
other's endowments. The remaining objects can be assigned
arbitrarily to the remaining agents, so
\[
|\TTR(P)\cap\T2TR(P)|=(n-2)!.
\]

We next show that there is exactly one PR allocation. Pair rationality requires agents 1 and 2 to receive each other's endowments; otherwise, exchanging their endowments would make both weakly better off and at least one strictly better off. Once these assignments are fixed, agents 3 and 4, if present, rank each other's endowments first among the remaining objects. PR therefore requires them to exchange their endowments as well. Continuing in this way, PR requires agents $i$ and $i+1$ to exchange their endowments for every odd $i<n$. If $n$ is odd, agent $n$ receives $o_n$.

This allocation indeed satisfies PR. Agents within each pair already receive each other's endowments. For agents belonging to different pairs, the agent in the earlier pair prefers her assignment to the other agent's endowment. The same argument applies to a pair involving the unpaired agent $n$, if present. Hence, no two agents can exchange their endowments so that both are weakly better off and at least one is strictly better off. Therefore, $|\PR(P)|=1$, and the desired ratio follows.
\end{proof}

In \autoref{prop:sp-tr-implies-ir}, we showed that although TR is much weaker than IR, it implies IR when combined with SP. An analogous result does not hold for TTR and T2TR: even when combined with SP, their conjunction does not imply PR.

\begin{proposition}\label{prop:sp-ttr-t2tr-notimplies-pr}
  There exists a rule that satisfies SP, TTR, and T2TR but violates PR.
\end{proposition}

\begin{proof}
The rule $\phi^{\mathrm{TR}}$ in \autoref{ex:fails-tr-pr} in the Appendix satisfies SP, TTR, and T2TR but violates PR.
\end{proof}

Together, \autoref{prop:sp-tr-implies-ir} and \autoref{prop:sp-ttr-t2tr-notimplies-pr} suggest the following interpretation: when attention is confined to SP rules, TR simply boils down to a reformulation of IR; the conjunction of TTR and T2TR, however, is not a reformulation of PR but imposes a strictly weaker requirement.

\section{Main result}\label{sec:main-result}

The rationality axioms can leave substantial discretion at a given profile,
as \autoref{prop:tr-weakerthan-ir} and \autoref{prop:ttr-t2tr-weakerthan-pr}
illustrate. The next theorem shows that imposing these requirements throughout
the preference domain, together with SP, determines the rule completely.

\begin{theorem}\label{thm:characterization}
TTC is the unique rule that satisfies SP, TR, TTR, and T2TR.
\end{theorem}

The proof is involved and is therefore deferred to \autoref{sec:proof}. The independence of the axioms is established in the \hyperref[sec:indep-of-axioms]{Appendix}. We conclude this section by presenting and discussing two corollaries of \autoref{thm:characterization}.

Since IR implies TR, and PR implies both TTR and T2TR, every rule satisfying
SP, IR, and PR satisfies the hypotheses of \autoref{thm:characterization}.
Moreover, TTC satisfies IR and PR because it selects core allocations.
This yields the following characterization in terms of the full rationality
conditions for individuals and pairs.

\begin{corollary}\label{cor:irpr}
  TTC is the unique rule that satisfies SP, IR, and PR.
\end{corollary}

The characterization also has implications when feasible trades are restricted.
Kidney exchange provides a motivation for bounding cycle length: coordinating
the simultaneous operations in an exchange cycle creates logistical demands
that increase with its size \citep{Roth2005JET}.
We capture this restriction by placing a cap on cycle size in our model.
The next corollary shows that such a cap is incompatible with SP and our three
rationality axioms, even without
requiring efficiency among the feasible allocations.

An allocation is a \textbf{$s$-way allocation} if every cycle it
executes has size at most $s$. A rule satisfies this axiom if it selects an $s$-way allocation at every profile.

\begin{corollary}\label{cor:restricted-cycles}
For $1\leq s<n$, there is no $s$-way exchange rule that satisfies SP, TR, TTR, and T2TR.
\end{corollary}

\begin{proof}
By \autoref{thm:characterization}, any such rule would coincide with TTC.
Consider a profile at which agents $1,\ldots,s+1$ form an all-tops cycle and
every remaining agent ranks her own endowment first. TTC executes the cycle of
size $s+1$, contradicting the restriction to $s$-way allocations.
\end{proof}

By the same implications,
\autoref{cor:restricted-cycles} also rules out $s$-way exchange rules
satisfying SP, IR, and PR. For two-way exchange, the
obstruction already arises at a single profile with three agents. At $P^*$ in
\autoref{ex:ttr-t2tr}, T2TR requires each agent to receive one of her top two
objects, both of which belong to the other agents. Any allocation satisfying
T2TR must therefore execute a cycle of size three. Thus, no two-way allocation satisfies
T2TR at this profile, and hence none satisfies PR. This particular obstruction
does not depend on SP or IR.

\section{Proof of Theorem~\ref{thm:characterization}}\label{sec:proof}

We first introduce the concept of an elevated cycle. We then establish two lemmas concerning cycles, which become useful when showing \autoref{thm:characterization}.

An all-tops cycle is an \textbf{elevated cycle} at $P$ if, in addition, each agent in the cycle ranks her own endowment second. The size of an elevated cycle is, by definition, at least two.

Observe that under an IR allocation, an elevated cycle is either executed or every agent in the cycle receives her endowment. To see this, consider an unexecuted elevated cycle. Then at least one agent in the cycle does not receive her top choice. By IR, she must receive her endowment, which she ranks second. Since her endowment is the previous agent's top choice, that agent cannot receive her top choice either and therefore must also receive her own endowment. Repeating this argument backward around the cycle shows that every agent in the cycle receives her endowment.

We establish two lemmas for rules satisfying SP and TR, which we call the elevator lemmas. \hyperref[lem:elevator1]{Elevator Lemma~\ref*{lem:elevator1}}, presented below, shows that the existence of an unexecuted all-tops cycle implies the existence of an unexecuted elevated cycle of the same size.

\begin{lemma}[Elevator Lemma 1]\label{lem:elevator1}
Let $\phi$ be a rule that satisfies TR and SP. For some profile $P$, if $\phi(P)$ does not execute an all-tops cycle of size $s$ at $P$, then there exists a profile $\hat{P}$ such that $\phi(\hat{P})$ does not execute an elevated cycle of size $s$ at $\hat{P}$.
\end{lemma}

\begin{proof}
By \autoref{prop:sp-tr-implies-ir}, $\phi$ satisfies IR.

Let $s$ denote the size of an all-tops cycle that $\phi(P)$ does not execute. Since $\phi$ satisfies TR, $s\geq2$. Relabeling agents if necessary, write this cycle as $(1,\dots,s)$ and suppose that $\phi_s(P)\neq o_1$.

For each agent $i\in\{1,\dots,s\}$, define $\hat P_i$ by placing her top choice first and her endowment second, while preserving the relative order of all remaining objects. Thus,
\[
\hat P_i:o_{i+1}\succ o_i\succ\dots
\quad\text{for }i<s,
\qquad
\hat P_s:o_1\succ o_s\succ\dots.
\]
Starting from $P$, change the reports of agents $s,s-1,\dots,1$ in this order. Let
\[
\hat P^{s+1}=P,
\qquad
\hat P^i=(\hat P_{\{i,\dots,s\}},P_{N\setminus\{i,\dots,s\}})
\quad\text{for }i=1,\dots,s.
\]

Consider first agent $s$. She cannot receive $o_1$ at $\hat P^s$: if her report changes, receiving $o_1$ would constitute a profitable deviation from $P$; if it does not, her assignment is unchanged. IR therefore implies that $\phi_s(\hat P^s)=o_s$. By feasibility, $\phi_{s-1}(\hat P^s)\neq o_s$.

More generally, suppose that agent $i<s$ does not receive her top choice, $o_{i+1}$, at $\hat P^{i+1}$. When her report changes to $\hat P_i$, SP rules out assigning her $o_{i+1}$, since this would be a profitable deviation from $\hat P^{i+1}$. IR therefore implies that $\phi_i(\hat P^i)=o_i$. If $i>1$, feasibility then gives $\phi_{i-1}(\hat P^i)\neq o_i$, establishing the premise for the next step.

Continuing down to agent $1$, we obtain $\phi_1(\hat P^1)=o_1$. At $\hat P=\hat P^1$, the cycle $(1,\dots,s)$ is elevated but is not executed. By IR and the observation preceding the lemma, every agent in the cycle receives her endowment.
\end{proof}

For the next lemma, we need to introduce the concepts of a conflict profile and an elevated conflict profile.

A profile $P$ is a \textbf{conflict profile} for rules $\phi$ and $\psi$ if $\phi(P)\neq\psi(P)$. Observe that if $\phi\neq\psi$, then such a profile exists.

A profile $P$ is an \textbf{elevated conflict profile} for rules $\phi$ and $\psi$ if there is an elevated cycle at $P$ that one of the allocations $\phi(P)$ and $\psi(P)$ executes and the other does not.

The main insight behind our next lemma comes from
\citet{Sethuraman2016ORL}, although the lemma is not stated as a separate result in that paper.\footnote{
\citet{Ekici2024EL} later uses the same insight to provide an alternative proof of the TTC characterization in \citet{Ekici2024TE}.}
\hyperref[lem:elevator2]{Elevator Lemma~\ref*{lem:elevator2}}, presented below, shows that, for two rules satisfying SP and TR, the existence of a conflict profile implies the existence of an elevated conflict profile.

\begin{lemma}[Elevator Lemma 2]\label{lem:elevator2}
Let $\phi\neq\psi$ be two rules that satisfy SP and TR. Then there exists an elevated conflict profile for $\phi$ and $\psi$.
\end{lemma}

\begin{proof}
By \autoref{prop:sp-tr-implies-ir}, both rules satisfy IR.

Call agent $i$ a \emph{conflict agent} at profile $P$ if $\phi_i(P)\neq\psi_i(P)$. For a conflict agent, we say that the \emph{top-two condition} is satisfied if her assignments under the two rules are her top choice and her endowment, which she ranks second. We first construct a conflict profile $P$ at which the top-two condition is satisfied for every conflict agent.

Since $\phi\neq\psi$, there exists a conflict profile $P^*$. If every conflict agent satisfies the top-two condition at $P^*$, set $P=P^*$; the construction is complete. Otherwise, choose a conflict agent $i$ for whom the condition is not satisfied. Without loss of generality, suppose that $\psi_i(P^*)\,P_i^*\,\phi_i(P^*)$. Since $\phi$ satisfies IR, $\phi_i(P^*)\,R_i^*\,o_i$. Hence, $\psi_i(P^*)\,P_i^*\,o_i$, so $\psi_i(P^*)\neq o_i$.

Consider the profile $(P_i,P_{-i}^*)$, where $P_i:\psi_i(P^*)\succ o_i\succ\dots$, and the unlisted objects retain their relative order from $P_i^*$. Since $\psi$ satisfies SP, $\psi_i(P_i,P_{-i}^*)=\psi_i(P^*)$. Since $\phi$ satisfies SP and IR, $\phi_i(P_i,P_{-i}^*)=o_i$. Thus, $(P_i,P_{-i}^*)$ remains a conflict profile, agent $i$ remains a conflict agent, but the top-two condition is now also satisfied for her.

If every conflict agent satisfies the top-two condition at $(P_i,P_{-i}^*)$, set $P=(P_i,P_{-i}^*)$. Otherwise, choose a conflict agent $j$ for whom the top-two condition is not satisfied. Applying the same argument, with the roles of the two rules interchanged if necessary, we get $(P_i,P_j,P_{N\setminus\{i,j\}}^*)$, which is another conflict profile but for which the top-two condition is now also satisfied for agent $j$.

Continue in this manner. Observe that once an agent's ranking has been changed, her endowment is ranked second, and the top-two condition continues to be satisfied for this agent whenever she is a conflict agent. Since each iteration yields a conflict profile, after at most $n$ iterations we obtain a
conflict profile $P$ at which the top-two condition is satisfied for every conflict agent.

We now show that $P$ is an elevated conflict profile for $\phi$ and $\psi$. Partition the set of agents into the sets $N_\phi$, $N_\psi$, and $\bar N$, where:
\begin{itemize}
    \item $N_\phi=\{i\in N:\phi_i(P)\,P_i\,\psi_i(P)\}$ is the set of conflict agents who prefer their assignments under $\phi(P)$. 

    \item $N_\psi=\{i\in N:\psi_i(P)\,P_i\,\phi_i(P)\}$ is the set of conflict agents who prefer their assignments under $\psi(P)$. 

    \item $\bar N=\{i\in N:\phi_i(P)=\psi_i(P)\}$ is the set of non-conflict agents.
\end{itemize}
By the top-two condition, each agent in $N_\phi$ receives her top choice under $\phi(P)$ and her endowment under $\psi(P)$, and each agent in $N_\psi$ receives her top choice under $\psi(P)$ and her endowment under $\phi(P)$.

Consider $i\in\bar N$, and suppose that $\phi_i(P)=\psi_i(P)=o_j$. If $j\in N_\phi$, then $j\neq i$ and $\psi_j(P)=o_j$, contradicting $\psi_i(P)=o_j$. If $j\in N_\psi$, then $j\neq i$ and $\phi_j(P)=o_j$, contradicting $\phi_i(P)=o_j$. Hence, $j\in\bar N$. Consequently, under both allocations, the agents in $\bar N$ receive exactly the endowments owned by members of $\bar N$. It follows that:
\begin{itemize}
    \item Under $\phi(P)$, the agents in $\bar N$ receive exactly their own set of endowments, and each agent in $N_\psi$ receives her endowment. The agents in $N_\phi$ therefore receive exactly the endowments of the agents in $N_\phi$.

    \item Under $\psi(P)$, the agents in $\bar N$ receive exactly their own set of endowments, and each agent in $N_\phi$ receives her endowment. The agents in $N_\psi$ therefore receive exactly the endowments of the agents in $N_\psi$.
\end{itemize}

Since $P$ is a conflict profile, at least one of $N_\phi$ and $N_\psi$ is nonempty. Without loss of generality, suppose that $N_\psi\neq\emptyset$. Under $\psi(P)$, the agents in $N_\psi$ receive exactly the endowments of agents in $N_\psi$, and none receives her own. Their assignments therefore contain a cycle $(i_1,\dots,i_s)$ of size $s\geq2$. Each agent in this cycle receives her top choice under $\psi(P)$ and ranks her endowment second, so it is an elevated cycle at $P$. It is executed by $\psi(P)$ but not by $\phi(P)$, proving that $P$ is an elevated conflict profile for $\phi$ and $\psi$.
\end{proof}

We are now ready to prove the theorem.

\begin{proof}[Proof of \autoref{thm:characterization}]
The TTC rule, $\varphi^{\mathrm{ttc}}$, satisfies SP \citep{Roth1982EL} and selects the unique core allocation at every profile \citep{Shapley1974JME,Roth1977JME}. A core allocation satisfies IR and PR, which imply TR, TTR, and T2TR. Hence, $\varphi^{\mathrm{ttc}}$ satisfies the four axioms in the theorem. It remains to show uniqueness.

Suppose, for contradiction, that $\phi\neq\varphi^{\mathrm{ttc}}$ satisfies SP, TR, TTR, and T2TR. By \autoref{prop:sp-tr-implies-ir}, $\phi$ satisfies IR.

By \hyperref[lem:elevator2]{Elevator Lemma~\ref*{lem:elevator2}}, there exists an elevated conflict profile for $\phi$ and $\varphi^{\mathrm{ttc}}$. Since TTC executes every elevated cycle, $\phi$ fails to execute an elevated cycle at some profile. Let $s$ be the minimum size of such a cycle, where the minimum is taken over all profiles. An elevated cycle has size at least two, and TTR requires $\phi$ to execute every elevated cycle of size two. Therefore, $s\geq3$.

In the rest of the proof the following two arguments will be useful:
\begin{itemize}
    \item \textbf{Smaller all-tops cycle argument.} At every profile, $\phi$ executes every all-tops cycle of size strictly less than $s$. Otherwise, \hyperref[lem:elevator1]{Elevator Lemma~\ref*{lem:elevator1}} would yield an unexecuted elevated cycle of size strictly less than $s$, contradicting the definition of $s$.

    \item \textbf{Trigger argument.} Consider distinct agents $i_1,\dots,i_{s-1}$ at a profile where each agent $i_r$, for $r=1,\dots,s-2$, ranks $o_{i_{r+1}}$ first, and agent $i_{s-1}$ ranks $o_{i_1}$ second. If agent $i_{s-1}$ changes her report to rank $o_{i_1}$ first, this \emph{triggers} the formation of the all-tops cycle $(i_1,\dots,i_{s-1})$. By the smaller all-tops cycle argument, $\phi$ executes this cycle, assigning her $o_{i_1}$. Since this report gives her her original second choice, SP implies that her assignment at the original profile is either her top choice or her second choice.
\end{itemize}

Choose a profile $P$ at which $\phi$ does not execute an elevated cycle of size $s$. Relabeling agents if necessary, write this cycle as $(1,\dots,s)$, and let $S=\{1,\dots,s\}$. By IR and the observation preceding \hyperref[lem:elevator1]{Elevator Lemma~\ref*{lem:elevator1}}, every agent in $S$ receives her endowment under $\phi(P)$.

In the remainder of the proof, we slightly abuse notation and interpret indices referring to agents in $S$ and their endowments cyclically modulo $s$. Thus, agent $s+1$ means agent $1$, and $o_{s+1}=o_1$ and $o_{s+2}=o_2$. This is not in line with our model specification, under which $s+1$ and $s+2$ denote distinct agents outside $S$ (if they exist), and $o_{s+1}$ and $o_{s+2}$ denote their endowments. However, this local notational convention is harmless: it simplifies the exposition without changing the substance of our arguments.

For each $i\in S$, define $\hat P_i$ by placing $o_{i+1}$ first, $o_{i+2}$ second, and $o_i$ third, while preserving the relative order of all remaining objects. Let $\hat P=(\hat P_S,P_{N\setminus S})$. The transformation is illustrated below. The boxes on the left indicate assignments under $\phi(P)$; we will show that the boxes on the right indicate assignments under $\phi(\hat P)$.
\[
P_S:\!
\left\{
\begin{aligned}
P_1 &: o_2\succ\boxed{o_1}\succ\dots\\
P_2 &: o_3\succ\boxed{o_2}\succ\dots\\
&\ \vdots\\
P_s &: o_1\succ\boxed{o_s}\succ\dots
\end{aligned}
\right.
\quad\longrightarrow\quad
\hat P_S:\!
\left\{
\begin{aligned}
\hat P_1 &: o_2\succ\boxed{o_3}\succ o_1\succ\dots\\
\hat P_2 &: o_3\succ\boxed{o_4}\succ o_2\succ\dots\\
&\ \vdots\\
\hat P_s &: o_1\succ\boxed{o_2}\succ o_s\succ\dots
\end{aligned}
\right.
\]

Change the reports of agents $1,\dots,s$ in this order. Define
\[
\hat P^0=P,
\qquad
\hat P^i=(\hat P_{\{1,\dots,i\}},P_{N\setminus\{1,\dots,i\}})
\quad\text{for }i=1,\dots,s.
\]
Thus, $\hat P^s=\hat P$. Every agent in $S$ retains the same top choice throughout this sequence.

We show inductively that, immediately before agent $i$ changes her report, she receives $o_i$, and immediately afterward, she receives $o_{i+2}$. For agent $1$, the first assertion follows from $\phi_1(P)=o_1$. Suppose now that $\phi_i(\hat P^{i-1})=o_i$. At $\hat P^i$, agent $i$ can create an all-tops cycle consisting of all agents in $S$ except agent $i+1$ by reporting her second choice, $o_{i+2}$, first. The trigger argument therefore implies that she receives either $o_{i+1}$ or $o_{i+2}$. By SP of $\phi$ we rule out $o_{i+1}$. Hence, $\phi_i(\hat P^i)=o_{i+2}$.

If $i<s$, agent $i+1$ still reports her original ranking, under which $o_{i+2}$ is first and $o_{i+1}$ is second. Since $o_{i+2}$ is assigned to agent $i$, by IR of $\phi$, we find that $\phi_{i+1}(\hat P^i)=o_{i+1}$.
This establishes the induction. In particular, $\phi_s(\hat P)=o_2$.

At $\hat P$, the trigger argument applies to every agent in $S$: by reporting her second choice first, each agent creates an all-tops cycle of size $s-1$. Thus, every agent in $S$ receives one of her top two choices under $\phi(\hat P)$. Since agent $s$ receives $o_2$, agent $1$ cannot receive her top choice and must receive $o_3$. Agent $2$ must then receive $o_4$, and the same reasoning propagates around the cycle. Consequently, we find that $\phi_i(\hat P)=o_{i+2}$ for every $i\in S$.

Finally, consider
\[
\tilde P=(\tilde P_2,\hat P_{-2}),
\qquad
\tilde P_2:o_3\succ o_1\succ\dots,
\]
where the unlisted objects retain their relative order from $\hat P_2$. By SP of $\phi$, agent $2$ does not receive $o_3$ at $\tilde P$. Agent $2$ ranks $o_1$ second at $\tilde P$, while agent $1$ ranks $o_2$ first. Since agent $2$ does not receive her top choice, by T2TR of $\phi$ we obtain that agents $1$ and $2$ exchange their endowments under $\phi(\tilde P)$, so $\phi_1(\tilde P)=o_2$ and $\phi_2(\tilde P)=o_1$. But then under $\phi(\tilde P)$ agent $s$ receives neither $o_1$ nor $o_2$, her top two choices. However, she can create the all-tops cycle $(s,2,\dots,s-1)$, of size $s-1$, by reporting $o_2$ first. The trigger argument therefore requires her to receive one of her top two choices at $\tilde P$, a contradiction. This establishes uniqueness.
\end{proof}
\section{Summary and related literature}\label{sec:sum-rellit}

In this paper, we have shown that SP, IR, and PR characterize TTC (\autoref{cor:irpr}). Our main result replaces the two rationality axioms with TR, TTR, and T2TR, showing that incentive compatibility together with rationality requirements involving only the first two ranks suffices to determine the unique core allocation.

Our paper's closest antecedents are axiomatic characterizations of TTC. \citet{Ma1994IJGT} characterizes TTC by SP, IR, and Pareto efficiency, for which \citet{Anno2015EL} and \citet{Sethuraman2016ORL} provide short proofs. \citet{Ekici2024TE} replaces Pareto efficiency with pair efficiency; \citet{Ekici2024EL} provide a short proof. In our characterization, we replace the efficiency requirement with PR. Note that these conditions address different exchanges: pair efficiency concerns exchanges of assignments, whereas PR concerns exchanges of endowments. They are logically incomparable, even among IR allocations (\autoref{ex:fails-t2tr} in the \hyperref[sec:indep-of-axioms]{Appendix}), so our result provides a distinct characterization. As in our paper, \citet{Fujinaka2018GEB} dispense with efficiency and characterize TTC by SP, IR, and endowments-swapping-proofness. This latter axiom precludes a pair of agents from benefiting by exchanging endowments before the rule is applied. \citet{Miyagawa2002aGEB} shows that SP, IR, anonymity, and nonbossiness characterize the class consisting of TTC and the no-trade rule. In relation to our axioms, observe that the no-trade rule fails TTR whenever two agents rank each other's endowments first.

The connection between incentive compatibility and the core also extends beyond object reallocation. \citet{Sonmez1999Econometrica} studies a general allocation framework and shows, under suitable domain assumptions, that SP, IR, and Pareto efficiency require a rule to select from the core whenever it is nonempty. Our paper takes a different route to core selection in the object reallocation model: we show that rationality requirements for individuals and pairs suffice when combined with SP, without imposing efficiency.

The role of endowments connects our analysis to a broader literature on allocation without transfers. \citet{Abdulkadiroglu1999JET} extend TTC to house allocation with existing tenants and newcomers, preserving SP, IR, and Pareto efficiency. \citet{Papai2000Econometrica} characterizes hierarchical exchange rules, which generalize TTC, by group strategy-proofness, Pareto efficiency, and reallocation-proofness, and \citet{Pycia2017TE} characterize the full class of group strategy-proof and Pareto-efficient rules as trading cycles. \citet{Papai2007JET} extends the exchange problem to agents endowed with multiple goods and characterizes fixed-deal exchange rules by SP, IR, and a weak efficiency requirement. \citet{Balbuzanov2019Econometrica} reconsider the interpretation of ownership, treating endowments as exclusion rights and developing the exclusion core. These studies broaden the allocation environment or the rights governing trade. In our paper, we retain the standard endowment structure and blocking notion, and ask how weak individual and bilateral rationality requirements can recover the core.

The comparison between DA and TTC has also been developed in school choice. \citet{Abdulkadiroglu2003AER} adapt both mechanisms to this setting, where school priorities replace ownership of objects. Their DA mechanism eliminates justified envy, while their TTC mechanism achieves Pareto efficiency. \citet{Morrill2013ET} provides axiomatic characterizations of TTC in this setting using independence of irrelevant rankings and mutual best. One of his characterizations also permits a comparison with DA. In our paper, we compare the two rules in the original marriage and object reallocation models: IR, PR, and strategy-proofness for the relevant agents identify the corresponding rule in each model.

Finally, our result on restrictions on cycle size connects the paper to kidney exchange. \citet{Roth2004QJE} adapt TTC to kidney exchange with preferences over compatible kidneys. The logistical demands of simultaneous operations motivate restrictions on cycle size, as discussed by \citet{Ashlagi2021MS}. Under dichotomous preferences, \citet{Roth2005JET} construct strategy-proof mechanisms that are efficient among feasible pairwise exchanges, while \citet{Roth2007AER} examine the gains from allowing three-way exchanges. Our corollary identifies an axiomatic obstacle on the unrestricted domain of strict preferences: a cycle cap below the number of agents is incompatible with SP, TR, TTR, and T2TR, even without requiring efficiency among feasible allocations (\autoref{cor:restricted-cycles}). Here, the difference in preference domains is essential: our result does not establish an impossibility for models in which patients are indifferent among compatible kidneys.

\section*{Appendix: Independence of axioms}
\label{sec:indep-of-axioms}

The independence of the four axioms in \autoref{thm:characterization} follows from Examples~\ref{ex:fails-sp}, \ref{ex:fails-tr-pr}, \ref{ex:fails-ttr}, and~\ref{ex:fails-t2tr}. The rules constructed in these examples are denoted by $\phi^{\mathrm{SP}}$, $\phi^{\mathrm{TR}}$, $\phi^{\mathrm{TTR}}$, and $\phi^{\mathrm{T2TR}}$, respectively. Each rule violates the axiom indicated by its superscript while satisfying the other three axioms in the theorem. Together, these examples show that none of the four axioms can be dropped from the characterization in general.

In the displayed profiles below, boxes indicate the assignments selected by the rule under consideration.

\begin{example}[\sout{SP}, TR, TTR, T2TR]
\label{ex:fails-sp}
Consider the case with three agents. Let $\mu^0=[o_1,o_2,o_3]$ denote the endowment allocation. Define the rule $\phi^{\mathrm{SP}}$ by
\[
\phi^{\mathrm{SP}}(P)=
\begin{cases}
\mu^0, & \text{if $\mu^0$ satisfies PR at $P$},\\
\varphi^{\mathrm{ttc}}(P), & \text{otherwise}.
\end{cases}
\]
Thus, the rule selects the endowment allocation whenever it satisfies PR and otherwise selects the TTC allocation.

\medskip
\noindent\textbf{TR, TTR, and T2TR.}
The rule satisfies IR because both the endowment allocation and the TTC allocation satisfy IR. It therefore satisfies TR. It also satisfies PR: the endowment allocation is selected only when it satisfies PR, and the TTC allocation always satisfies PR. Since PR implies TTR and T2TR, the rule satisfies both axioms.

\medskip
\noindent\textbf{Failure of SP.}
Consider the profile
\[
P:\!
\left\{
\begin{aligned}
P_1 &: o_2 \succ \boxed{o_1} \succ o_3,\\
P_2 &: o_3 \succ \boxed{o_2} \succ o_1,\\
P_3 &: o_1 \succ \boxed{o_3} \succ o_2.
\end{aligned}
\right.
\]
The endowment allocation satisfies PR at $P$: no two agents both prefer the other's endowment to their own. Hence, $\phi^{\mathrm{SP}}(P)=\mu^0$.

Now suppose agent 3 reports
\[
\tilde P_3:o_1\succ o_2\succ o_3.
\]
At $\tilde P=(\tilde P_3,P_{-3})$, the endowment allocation violates PR because agents 2 and 3 would both benefit from exchanging their endowments. The rule therefore selects the TTC allocation. Since the agents' top choices still form the cycle $(1,2,3)$,
\[
\phi^{\mathrm{SP}}(\tilde P)=[o_2,o_3,o_1].
\]
Agent 3 receives $o_1$ instead of $o_3$, a strict improvement according to her true preferences. Thus, the rule violates SP.
\end{example}

\begin{example}[SP, \sout{TR}, TTR, T2TR]
\label{ex:fails-tr-pr}
Consider the case with three agents. Define the rule $\phi^{\mathrm{TR}}$ by \autoref{tab:fails-tr-pr}. The rule checks the rows from top to bottom and selects the allocation specified in the first row whose conditions are satisfied.

\begin{table}[htbp]
\centering
\caption{A rule that satisfies SP, TTR, and T2TR but violates TR and PR}
\label{tab:fails-tr-pr}
\renewcommand{\arraystretch}{1.5}
\setlength{\tabcolsep}{8pt}
\begin{tabular}{
    p{0.30\textwidth}
    p{0.42\textwidth}
    p{0.17\textwidth}
}
\hline
Agent 1's preference
    & Additional condition
    & \centering $\phi^{\mathrm{TR}}(P)$\tabularnewline
\hline

$P_1\colon o_1\succ o_2\succ o_3$
    & $P_2\colon o_2\succ o_1\succ o_3$
    & \centering $[o_2,o_3,o_1]$\tabularnewline

$P_1\colon o_1\succ o_3\succ o_2$
    & $P_2\colon o_1\succ o_2\succ o_3$ or $P_2\colon o_2\succ o_3\succ o_1$
    & \centering $[o_1,o_2,o_3]$\tabularnewline

$P_1\colon o_1\succ o_3\succ o_2$
    & $P_2\colon o_2\succ o_1\succ o_3$
    & \centering $[o_3,o_2,o_1]$\tabularnewline

$o_2$ is agent 1's top choice
    & $o_1\,P_2\,o_3$
    & \centering $[o_2,o_1,o_3]$\tabularnewline

$o_2$ is agent 1's top choice
    & $o_3\,P_2\,o_1$ and $o_1\,P_3\,o_2$
    & \centering $[o_2,o_3,o_1]$\tabularnewline

$o_3$ is agent 1's top choice
    & $o_1\,P_3\,o_2$
    & \centering $[o_3,o_2,o_1]$\tabularnewline

$o_3$ is agent 1's top choice
    & $o_2\,P_3\,o_1$ and $o_1\,P_2\,o_3$
    & \centering $[o_3,o_1,o_2]$\tabularnewline

Any remaining case
    & ---
    & \centering $[o_1,o_3,o_2]$\tabularnewline
\hline
\end{tabular}
\end{table}

\medskip
\noindent\textbf{SP.}
Fixing the other agents' reports, call the set of objects an agent can obtain through her own reports her \emph{menu}. It suffices to show that truthful reporting always selects her most-preferred object from this menu.

For agent 1, there are three cases. If $P_2:o_2\succ o_1\succ o_3$, her menu is $\{o_2,o_3\}$, and the rule assigns whichever of these she prefers. If $o_3\,P_2\,o_1$ and $o_2\,P_3\,o_1$, her menu is the singleton $\{o_1\}$. In every remaining case, her menu is $O$, and she receives her reported top choice.

For agent 2, if $P_1:o_1\succ o_2\succ o_3$, her assignment is always $o_3$. If $P_1:o_1\succ o_3\succ o_2$, her menu is $\{o_2,o_3\}$, and she receives whichever of these she prefers. If agent 1 ranks $o_2$ first, agent 2's menu is $\{o_1,o_3\}$, and she again receives her most-preferred object in her menu. Finally, if agent 1 ranks $o_3$ first, agent 2 always receives $o_2$ when $o_1\,P_3\,o_2$; otherwise, she receives her most-preferred object in $\{o_1,o_3\}$.

For agent 3, her assignment is independent of her report whenever agent 1 ranks $o_1$ first, or whenever agent 1 ranks $o_2$ first and $o_1\,P_2\,o_3$. In every remaining case, her menu is $\{o_1,o_2\}$, and she receives whichever of these she prefers. Thus, no agent can benefit from misreporting, establishing SP.

\medskip
\noindent\textbf{TTR.}
If agents 1 and 2 rank each other's endowments first, the fourth row selects $[o_2,o_1,o_3]$. If agents 1 and 3 do so, the sixth row selects $[o_3,o_2,o_1]$. If agents 2 and 3 do so, none of the first seven rows applies, and the last row selects $[o_1,o_3,o_2]$. Thus, the rule satisfies TTR.

\medskip
\noindent\textbf{T2TR.}
Consider an ordered pair $(i,j)$ such that agent $i$ ranks $o_j$ second and agent $j$ ranks $o_i$ first. There are six cases:
\begin{itemize}
    \item \textbf{Case 1:} $(i,j)=(1,2)$. Agent 1's top choice is either $o_1$ or $o_3$. In the former case, the last row assigns her $o_1$; in the latter, the sixth or seventh row assigns her $o_3$. Thus, she always receives her top choice.

    \item \textbf{Case 2:} $(i,j)=(1,3)$. If agent 1 ranks $o_2$ first, the fourth or fifth row assigns it to her. If she ranks $o_1$ first, she receives it except when $P_2:o_2\succ o_1\succ o_3$; in that case, the third row selects $[o_3,o_2,o_1]$.

    \item \textbf{Case 3:} $(i,j)=(2,1)$. Agent 2's ranking is either $o_2\succ o_1\succ o_3$ or $o_3\succ o_1\succ o_2$. In the former case, the fourth row selects $[o_2,o_1,o_3]$. In the latter, the fifth or last row assigns agent 2 her top choice, $o_3$.

    \item \textbf{Case 4:} $(i,j)=(2,3)$. Agent 2's ranking is either $o_1\succ o_3\succ o_2$ or $o_2\succ o_3\succ o_1$. In the former case, she receives her top choice when agent 1 ranks $o_2$ or $o_3$ first. In the latter, she receives her top choice when $P_1:o_1\succ o_3\succ o_2$. In all remaining cases, the last row selects $[o_1,o_3,o_2]$.

    \item \textbf{Case 5:} $(i,j)=(3,1)$. Agent 3's ranking is either $o_2\succ o_1\succ o_3$ or $o_3\succ o_1\succ o_2$. In the former case, the seventh or last row assigns her $o_2$, her top choice. In the latter, the sixth row selects $[o_3,o_2,o_1]$.

    \item \textbf{Case 6:} $(i,j)=(3,2)$. Agent 3's ranking is either $o_1\succ o_2\succ o_3$ or $o_3\succ o_2\succ o_1$. In the former case, she receives her top choice whenever agent 1 ranks $o_2$ or $o_3$ first; otherwise, the last row selects $[o_1,o_3,o_2]$. In the latter case, the last row always selects $[o_1,o_3,o_2]$.
\end{itemize}
Thus, in every case, either agent $i$ receives her top choice or agents $i$ and $j$ exchange their endowments. This establishes T2TR.

\medskip
\noindent\textbf{Failure of TR.}
Consider the profile
\[
P:\!
\left\{
\begin{aligned}
P_1 &: o_1 \succ \boxed{o_2} \succ o_3,\\
P_2 &: o_2 \succ o_1 \succ \boxed{o_3},\\
P_3 &: \boxed{o_1} \succ o_2 \succ o_3.
\end{aligned}
\right.
\]
The rule selects $\phi^{\mathrm{TR}}(P)=[o_2,o_3,o_1]$. Agent 1 ranks her endowment first but receives $o_2$. Thus, the rule violates TR.

The rule also violates PR. At the same profile, agents 1 and 2 could exchange their endowments and receive $o_2$ and $o_1$, respectively. Agent 1 would retain her assignment, while agent 2 would be strictly better off. This also establishes that SP, TTR, and T2TR together do not imply PR.
\end{example}

\begin{example}[SP, TR, \sout{TTR}, T2TR]
\label{ex:fails-ttr}
Consider the case with two agents. Let $\phi^{\mathrm{TTR}}$ denote the endowment rule, which assigns each agent her endowment: $\phi^{\mathrm{TTR}}(P)=[o_1,o_2]$ for every profile $P$.

The rule satisfies SP because its outcome is independent of the agents' reports. It satisfies IR and therefore TR because each agent receives her endowment. It also satisfies T2TR: whenever agent $i$ ranks agent $j$'s endowment second, she ranks her own endowment first and therefore receives her top choice.

However, the rule violates TTR at the profile
\[
P:\!
\left\{
\begin{aligned}
P_1 &: o_2 \succ \boxed{o_1},\\
P_2 &: o_1 \succ \boxed{o_2}.
\end{aligned}
\right.
\]
The rule selects $\phi^{\mathrm{TTR}}(P)=[o_1,o_2]$. Both agents rank each other's endowments first, but neither receives her top choice. Thus, the rule violates TTR.
\end{example}

\begin{example}[SP, TR, TTR, \sout{T2TR}]
\label{ex:fails-t2tr}
Consider the case with three agents. Define the rule $\phi^{\mathrm{T2TR}}$ as follows.

At a given profile, let each agent point to her top-ranked object and each object point to its owner. If there is a cycle of size one or two, the rule selects the TTC allocation. Otherwise, the agents form a three-way cycle. In this case, each agent ranks her endowment either second or last. Call these agents \emph{type S} and \emph{type L}, respectively.

For a three-way cycle, the rule assigns objects as follows:
\begin{itemize}
    \item \textbf{Case 1: All agents have the same type.} Assign each agent her second choice.

    \item \textbf{Case 2: Two agents are type L and one is type S.} The two type-L agents exchange their endowments, and the type-S agent receives her endowment.

    \item \textbf{Case 3: Two agents are type S and one is type L.} The type-S agent who points to the other type-S agent’s endowment receives her endowment. The remaining two agents exchange their endowments.
\end{itemize}

These assignments are feasible. In Case 1, assigning second choices reverses the three-way cycle when all agents are type L and gives the endowment allocation when all agents are type S. In Cases 2 and 3, two agents exchange their endowments and the remaining agent receives her endowment.

The following observation will be useful for verifying SP. For any agent $k$, write the three-way cycle as $(k,i,j)$. The rule assigns $k$ her top choice if $i$ is type L and $j$ is type S; otherwise, it assigns $k$ her second choice. Thus, whether $k$ receives her first or second choice depends only on the other two agents' types.

\medskip
\noindent\textbf{TR and TTR.}
Whenever the rule selects the TTC allocation, IR holds. In a three-way cycle, every agent receives her first or second choice. A type-S agent ranks her endowment second, and a type-L agent ranks her endowment last. Hence, IR holds in this case as well. The rule therefore satisfies TR.

If two agents rank each other's endowments first, they form a cycle of size two. The rule then selects the TTC allocation, which executes their cycle. Thus, the rule satisfies TTR.

\medskip
\noindent\textbf{SP.}
Fix an agent $k$ and the reports of the other two agents, $i$ and $j$.

If $i$ or $j$ points to her own endowment, or $i$ and $j$ point to each other’s endowments, the rule coincides with TTC regardless of $k$'s report. Strategy-proofness of TTC therefore rules out a profitable deviation by $k$.

Otherwise, every cycle must contain $k$. If truthful reporting places $k$ in a cycle of size one or two, she receives her top choice and cannot benefit from a deviation. It remains to consider the case in which truthful reporting produces a three-way cycle, $(k,i,j)$.

If $k$ changes her reported top choice, she points either to $o_k$ or to $o_j$. This creates a cycle of size one or two, respectively. The rule then selects the TTC allocation and assigns $k$ her reported top choice, which is her true second or third choice. Since truthful reporting gives her either her true first or her true second choice, such a deviation cannot benefit her.

If $k$ keeps her top choice unchanged, the only possible misreport reverses her second and third choices. By the observation above, whether she receives her reported first or second choice depends only on the types of $i$ and $j$. If $i$ is type L and $j$ is type S, she receives her top choice under either report. Otherwise, truthful reporting gives her her true second choice, whereas the misreport gives her her true third choice. Thus, this deviation cannot benefit her either. This establishes SP.

\medskip
\noindent\textbf{Failure of T2TR.}
Consider the profile
\[
P:\!
\left\{
\begin{aligned}
P_1 &: o_2 \succ \boxed{o_3} \succ o_1,\\
P_2 &: o_3 \succ \boxed{o_1} \succ o_2,\\
P_3 &: o_1 \succ \boxed{o_2} \succ o_3.
\end{aligned}
\right.
\]
The agents form a three-way cycle and are all type L. Therefore, the rule assigns each agent her second choice: $\phi^{\mathrm{T2TR}}(P)=[o_3,o_1,o_2]$.

Agent 1 ranks $o_3$ second, and agent 3 ranks $o_1$ first. However, agent 1 does not receive her top choice, and agents 1 and 3 do not exchange their endowments: agent 3 receives $o_2$ instead of $o_1$. Thus, the rule violates T2TR.

\medskip
\noindent\textbf{Pair rationality and pair efficiency.}
This example also shows that PR and pair efficiency are logically independent, even among individually rational allocations. At the displayed profile, the allocation $[o_3,o_1,o_2]$ is pair efficient because every exchange of assignments makes one participant worse off. It violates PR: agents 1 and 3 can exchange their endowments, leaving agent 1's assignment unchanged while strictly improving agent 3's assignment from $o_2$ to $o_1$.

For the converse, consider Case 3 and relabel the agents so that the three-way cycle is $(1,2,3)$, agents 1 and 2 are type S, and agent 3 is type L. The rule selects $[o_1,o_3,o_2]$. Agents 2 and 3 already exchange their endowments. An exchange of endowments between agents 1 and 2 would make agent 2 worse off, and one between agents 1 and 3 would make agent 1 worse off. Thus, PR holds. However, agents 1 and 3 can both obtain their top choices by exchanging their assignments, so pair efficiency fails. Both allocations satisfy IR, as established above for every allocation selected by the rule.
\end{example}

\bibliography{matchingBib}

\end{document}